\documentclass[letterpaper,12pt,leqno]{article}
\usepackage{paper}
\hypersetup{pdftitle={A Sequential Foundation for Two-Level Nested Logit}}
\newcommand{\bib}{references.bib}

\newcommand{\Ihg}{I_{hg}}
\newcommand{\Ig}{I_{g}}

\begin{document}

\title{\Large A Sequential Foundation for Two-Level Nested Logit}
\author{Davide Luparello\thanks{\href{mailto:dluparello@fsu.edu}{dluparello@fsu.edu}.}}
\affil{{\normalsize\it Department of Economics, Florida State University, 278 Bellamy Building,}\\{\normalsize\it 113 Collegiate Loop, Tallahassee, FL 32306, USA}\\{\normalsize\it Baffi Centre, Bocconi University, Via R{\"o}ntgen 1, 20136 Milan, Italy}}
\date{\today}

\begin{titlepage}\maketitle
\begin{abstract}
I show that the two-level nested logit model arises exactly from a sequential choice process in which group, subgroup, and product shocks are revealed in stages. The result extends to any number of levels and holds for any nesting parameters below one, without the bounds and ordering that random utility maximization requires. Parameter estimates that applied work discards for violating the standard ordering or the lower bound are consistent with sequential choice. The linear inversion and welfare formula hold. When the standard ordering fails, improving one product can raise a rival's market share in its subgroup, even with single-unit demand.
\end{abstract}


\end{titlepage}

\section{Introduction}\label{sec:intro}

The two-level nested logit is a standard demand model for markets segmented along two dimensions, and its parameters are typically estimated through the linear inversion of \citet{Berry1994} and \citet{Verboven1996}. The model is globally consistent with random utility maximization only if the nesting parameters for products within a subgroup, $\sigma_1$, and for subgroups within a group, $\sigma_2$, satisfy $0 \le \sigma_2 \le \sigma_1 < 1$ \citep{McFadden1978,McFadden1981,Cardell1997}, and applied work uses this restriction to select the nesting order, discarding trees whose estimates fall outside it.\footnote{See \citet{Goldberg1995}, \citet{Verboven1996}, and \citet{GrigolonVerboven2014} for the car market; \citet[p.~319]{BenAkivaLerman1985} reject, for the same reason, a better-fitting decision tree for travel demand.}

In this paper, I show that the two-level nested logit share system and its linear inversion arise exactly from a sequential choice process in which group, subgroup, and product taste shocks are independent extreme value draws, revealed one stage at a time. At each stage, the consumer compares nests through their inclusive values, the expected utility, net of the shocks already observed, of choosing optimally within each nest once the taste shocks within it are revealed. The process is well defined for any $\sigma_1, \sigma_2 < 1$, whatever their order, and the result extends to any number of levels.

The equivalence has three implications. First, the estimating equation of \citet{Berry1994} and \citet{Verboven1996} is unchanged, so existing estimates outside the restricted parameter region can be read as estimates of the sequential process. Second, ex ante expected utility equals the nested logit inclusive value, so the standard welfare formula holds for any $\sigma_1, \sigma_2 < 1$. Third, when $\sigma_2 > \sigma_1$, raising the mean utility of one product can raise the choice probability of another in the same subgroup.

Nested logit has long been read as a hierarchy of choices. \citet{Williams1977} represents it as a utility tree in which the consumer, with full information, compares the realized maximum utility of each nest. \citet{BenAkiva1973} reads logsum coefficients different from one as evidence of sequential choice but assumes a logit model at each stage rather than deriving it. \citet{Rust1987} and \citet{Train2003} obtain logsum terms from shocks revealed over time, where the logsum coefficient weights future utility. \citet[note~6]{TrainMcFaddenBenAkiva1987} observe that an inclusive value coefficient above one may reflect dynamic aspects of choice outside the random utility model; the sequential process is one such account. \citet[p.~356, note~4]{DalyZachary1978} note that the static model admits no representation of the maximum utility within a nest as its inclusive value plus a random term unaffected by the nest's mean utilities. In the sequential process, the consumer values a nest, before observing the shocks within it, at its inclusive value plus an independent shock, from which I derive a logit model at each stage and nested logit for any $\sigma_1, \sigma_2 < 1$.

Some values outside the restricted parameter region, including $\sigma_2 > \sigma_1$, are consistent with random utility maximization at the values of the data, though not globally \citep{BorschSupan1990,HerrigesKling1996,GilMoltoHole2004}. The sequential process is consistent with expected utility maximization for any $\sigma_1, \sigma_2 < 1$. \citet{KovachTserenjigmid2022} read negative nesting parameters as complementarity among products, and \citet{FosgerauMonardoDePalma2024} obtain complementarity in demand from nonhierarchical groupings, with parameters inside the restricted parameter region. The sequential process generates positive cross effects within a subgroup with nested groupings and single-unit demand.

The sequential process can also be read as a model of consideration sets \citep[e.g.,][]{Goeree2008}: the consumer considers only the products of the subgroup that maximizes expected utility. All cross effects of mean utilities on choice probabilities nevertheless remain symmetric, as in nested logit, so the asymmetric demand responses that \citet{AbaluckAdamsPrassl2021} use to identify limited consideration do not arise. \citet{DuboisSaethre2020} estimate pharmaceutical demand with unobserved choice sets and sequentially observed pharmacy-chain and product shocks. I extend this timing to multiple levels, with deterministic product availability.

\section{The sequential choice process}\label{sec:model}

Products $j$ are partitioned into subgroups $h \in H_g$ within groups $g \in G$; $J_{hg}$ denotes the products in subgroup $h$ of group $g$, and the outside option is $j=0$. The utility of consumer $i$ from product $j \in J_{hg}$ is
\begin{equation}\label{eq:utility}
U_{ij} = \delta_j + \zeta_i^{g} + (1-\sigma_2)\,\zeta_i^{h} + (1-\sigma_1)\,\epsilon_{ij},
\end{equation}
and the utility from the outside option is $U_{i0} = \zeta_i^{0}$. Here $\delta_j$ is the mean utility of product $j$; $\zeta_i^{0}$, $\zeta_i^{g}$, $\zeta_i^{h}$, and $\epsilon_{ij}$ are consumer $i$'s shocks to the outside option, group $g$, subgroup $h$, and product $j$; and $\sigma_1$ and $\sigma_2$ are the subgroup and group nesting parameters. Equation~\eqref{eq:utility} has a form similar to the specifications of \citet{Cardell1997}, \citet{Verboven1996}, and \citet{MansleyMillerRyanWeinberg2019}; the static and sequential processes differ only in the distribution of the shocks and in what the consumer knows when each decision is made.

\begin{assumption}\label{ass:shocks}
The shocks $\{\zeta_i^{0}, \zeta_i^{g}\}_{g\in G}$, $\{\zeta_i^{h}\}_{h\in H_g, g\in G}$, and $\{\epsilon_{ij}\}_{j}$ are mutually independent extreme value draws with unit scale and mean zero, and $\sigma_1, \sigma_2 < 1$.
\end{assumption}

\begin{assumption}\label{ass:timing}
The consumer observes the shocks $\{\zeta_i^{0}, \zeta_i^{g}\}_{g\in G}$ and chooses between the outside option and a group $g$; then observes the subgroup shocks $\{\zeta_i^{h}\}_{h\in H_g}$ and chooses a subgroup $h$; then observes the product shocks $\{\epsilon_{ij}\}_{j\in J_{hg}}$ and chooses a product. At each stage, the consumer maximizes expected utility given the shocks observed so far.
\end{assumption}

Since the shocks have unit scale, $1-\sigma_2$ and $1-\sigma_1$ measure taste heterogeneity across subgroups and across products, relative to that across groups. Under the static random utility model, the shocks $\zeta_i^{g}$ and $\zeta_i^{h}$ instead follow the conjugate distributions of \citet{Cardell1997}, which exist only if $0 \le \sigma_2 \le \sigma_1 < 1$. Assumption~\ref{ass:shocks} requires only $\sigma_1, \sigma_2 < 1$.\footnote{The mean-zero normalization makes $\delta_j$ the mean utility of product $j$ relative to the outside option. Under standard extreme value shocks, this relative mean utility is $\delta_j$ plus $(2-\sigma_1-\sigma_2)$ times Euler's constant for every inside product, and an intercept absorbs the constant.} Define the inclusive values
\begin{equation}\label{eq:inclusive}
\Ihg = (1-\sigma_1)\log\sum_{j\in J_{hg}} \exp\!\left(\frac{\delta_j}{1-\sigma_1}\right),
\qquad
\Ig = (1-\sigma_2)\log\sum_{h\in H_g} \exp\!\left(\frac{\Ihg}{1-\sigma_2}\right).
\end{equation}

\begin{proposition}\label{prop:equivalence}
Under Assumptions~\ref{ass:shocks} and~\ref{ass:timing}, for any $\sigma_1, \sigma_2 < 1$, the probability that consumer $i$ chooses product $j \in J_{hg}$ is $s_j = s_{j|hg}\, s_{h|g}\, s_g$, with
\begin{equation}\label{eq:shares}
s_{j|hg} = \frac{\exp\left(\delta_j/(1-\sigma_1)\right)}{\exp\left(\Ihg/(1-\sigma_1)\right)},
\qquad
s_{h|g} = \frac{\exp\left(\Ihg/(1-\sigma_2)\right)}{\exp\left(\Ig/(1-\sigma_2)\right)},
\qquad
s_g = \frac{\exp(\Ig)}{1+\sum_{g'\in G}\exp(I_{g'})}.
\end{equation}
These are the two-level nested logit shares. Consequently,
\begin{equation}\label{eq:inversion}
\log s_j - \log s_0 = \delta_j + \sigma_1 \log s_{j|hg} + \sigma_2 \log s_{h|g},
\end{equation}
and the consumer's ex ante expected utility is $W = \log\big(1+\sum_{g\in G}\exp(\Ig)\big)$.
\end{proposition}

\begin{proof}
I solve the consumer's problem by backward induction, using the fact that if $\{e_k\}$ are independent extreme value draws with unit scale and mean zero, then $\max_k\{a_k + \lambda e_k\}$ for $\lambda > 0$ has mean $\lambda\log\sum_k \exp(a_k/\lambda)$, and the probability that $k$ attains the maximum is $\exp(a_k/\lambda)/\sum_{k'}\exp(a_{k'}/\lambda)$ \citep[Theorem 5.2]{McFadden1981}.

At the product stage, the consumer has observed $\zeta_i^{g}$ and $\zeta_i^{h}$ and chooses $j \in J_{hg}$ to maximize $U_{ij}$; since $\zeta_i^{g} + (1-\sigma_2)\zeta_i^{h}$ is common to all $j \in J_{hg}$, this is equivalent to maximizing $\delta_j + (1-\sigma_1)\epsilon_{ij}$. With $\lambda = 1-\sigma_1 > 0$, the expected maximum is $\Ihg$, and the choice probability is $s_{j|hg}$ in \eqref{eq:shares}. At the subgroup stage, the consumer chooses $h \in H_g$ to maximize $\Ihg + \zeta_i^{g} + (1-\sigma_2)\zeta_i^{h}$; since $\zeta_i^{g}$ is common to all $h \in H_g$, this is equivalent to maximizing $\Ihg + (1-\sigma_2)\zeta_i^{h}$. With $\lambda = 1-\sigma_2 > 0$, the expected maximum is $\Ig$ and the choice probability is $s_{h|g}$. The value of group $g$ at the first stage is $V_{ig} = \Ig + \zeta_i^{g}$, and the value of the outside option is $\zeta_i^{0}$. With $\lambda = 1$, the expected maximum is $W$, and the choice probability is $s_g$.

For the inversion, $s_0 = 1/(1+\sum_{g'}\exp(I_{g'}))$, so $\log s_g - \log s_0 = \Ig$. From \eqref{eq:shares}, $\Ig = \Ihg - (1-\sigma_2)\log s_{h|g}$ and $\Ihg = \delta_j - (1-\sigma_1)\log s_{j|hg}$. Adding $\log s_{j|hg} + \log s_{h|g}$ to $\log s_g - \log s_0$ and substituting yields \eqref{eq:inversion}, as in \citet{Berry1994} for one level, extended to two levels by \citet{Verboven1996}.
\end{proof}

The proof relies on i.i.d.\ extreme value shocks at every level, and its only restriction on the parameters is $\sigma_1, \sigma_2 < 1$. In the static model, the composite shock $\zeta_i^{h} + \frac{1-\sigma_1}{1-\sigma_2}\epsilon_{ij}$ must be extreme value, which is possible only if $\frac{1-\sigma_1}{1-\sigma_2} \le 1$ \citep[i.e., $\sigma_2 \le \sigma_1$; see][Theorem 2.1]{Cardell1997}. In the sequential process, the consumer integrates over the product shocks before observing them, so the two shocks never need to combine into a single extreme value variable.

The argument does not depend on the number of levels. With $K \ge 1$ nesting levels, product $j$ belongs to nests $n_1(j) \subset \cdots \subset n_K(j)$; let $n_0(j) = j$ (the product itself), $\zeta_i^{j} = \epsilon_{ij}$, and $\sigma_{K+1} \equiv 0$. The utility of consumer $i$ from product $j$ is
\begin{equation}\label{eq:utilityK}
U_{ij} = \delta_j + \sum_{k=1}^{K} (1-\sigma_{k+1})\,\zeta_i^{n_k(j)} + (1-\sigma_1)\,\epsilon_{ij},
\end{equation}
which is \eqref{eq:utility} when $K=2$. The consumer observes the shocks one level at a time, from the largest nests down to the products. Let $I_j = \delta_j$ for products, $I_0 = 0$ for the outside option, and $n_{K+1}(j)$ be the universal nest, whose members are the outside option and the nests at level $K$. For a nest $n$ at level $k \le K+1$ whose members (nests at level $k-1$, or products if $k=1$) form the set $C(n)$, define the inclusive value of $n$
\begin{equation}\label{eq:inclusiveK}
I_n = (1-\sigma_k)\log\sum_{m\in C(n)}\exp\!\left(\frac{I_m}{1-\sigma_k}\right).
\end{equation}

\begin{proposition}\label{prop:levels}
Under Assumption~\ref{ass:shocks} for $\zeta_i^{0}$ and the shocks $\zeta_i^{n}$ of all nests and products, and Assumption~\ref{ass:timing} with one stage per level, from level $K$ down to the products, for any $\sigma_1,\dots,\sigma_K < 1$, the choice probability of product $j$ is
\begin{equation}\label{eq:sharesK}
s_j = \prod_{k=1}^{K+1} s_{n_{k-1}(j)|n_k(j)},
\qquad
s_{m|n} = \exp\!\left(\frac{I_m - I_n}{1-\sigma_k}\right) \text{ for } m \in C(n),\ n \text{ at level } k.
\end{equation}
These are the $K$-level nested logit shares. Consequently,
\begin{equation}\label{eq:inversionK}
\log s_j - \log s_0 = \delta_j + \sum_{k=1}^{K} \sigma_k \log s_{n_{k-1}(j)|n_k(j)},
\end{equation}
and the consumer's ex ante expected utility is $W = I_{n_{K+1}(j)}$.
\end{proposition}

\begin{proof}
The proof is by backward induction on the level. Let $n$ be a nest at level $k$, and suppose that the expected utility of choosing $m \in C(n)$, net of the shocks already observed, is $I_m$, as holds for products since $I_j = \delta_j$. The shocks $\zeta_i^{m}$ enter $U_{ij}$ with coefficient $1-\sigma_k$, and the shocks of $n$ and the nests above it are common to all $m \in C(n)$, so the consumer maximizes $I_m + (1-\sigma_k)\zeta_i^{m}$. With $\lambda = 1-\sigma_k > 0$, the expected maximum is $I_n$, which completes the induction, and the choice probability is $s_{m|n}$. At the first stage, $n = n_{K+1}(j)$ and $\lambda = 1$, so $W = I_{n_{K+1}(j)}$. For the inversion, $s_{m|n}$ with $I_0 = 0$ gives $\log s_{n_K(j)|n_{K+1}(j)} - \log s_0 = I_{n_K(j)}$, so $\log s_j - \log s_0 = I_{n_K(j)} + \sum_{k=1}^{K}\log s_{n_{k-1}(j)|n_k(j)}$. Taking logs in $s_{m|n}$ gives $I_{n_k(j)} = I_{n_{k-1}(j)} - (1-\sigma_k)\log s_{n_{k-1}(j)|n_k(j)}$ for $k = 1,\dots,K$; summing over $k$ with $I_{n_0(j)} = \delta_j$ gives $I_{n_K(j)} = \delta_j - \sum_{k=1}^{K}(1-\sigma_k)\log s_{n_{k-1}(j)|n_k(j)}$, and substituting yields \eqref{eq:inversionK}.
\end{proof}

Applied level by level, the construction of \citet{Cardell1997} requires $0 \le \sigma_K \le \cdots \le \sigma_1 < 1$. The sequential process allows any ordering at any depth.

\section{Implications}\label{sec:implications}

\paragraph{Estimation and interpretation} Estimates of $\sigma_1$ and $\sigma_2$ can be read as estimates of the sequential process, whether or not they satisfy $0\le\sigma_2\le\sigma_1$. Under this process, $\sigma_2 > \sigma_1$ means that tastes vary more across products within a subgroup than across subgroups within a group, and a negative $\sigma_1$ or $\sigma_2$ means that tastes vary more across products or across subgroups, respectively, than across groups. Inside the restricted parameter region, the static and sequential processes produce identical choice probabilities, so data on final choices cannot distinguish them, although data on when consumers observe their shocks could; outside it, the sequential process still gives the estimates an expected-utility-maximizing interpretation.

\paragraph{Cross effects within a subgroup} Let $\beta = (1-\sigma_1)/(1-\sigma_2)$. Differentiating the shares in \eqref{eq:shares}, for two distinct products $j, k \in J_{hg}$,
\begin{equation}\label{eq:cross}
\frac{\partial s_j}{\partial \delta_k} = \frac{s_j\, s_{k|hg}}{1-\sigma_1}\Big[(\beta - 1) - \beta\, s_{h|g}\big(\sigma_2 + (1-\sigma_2)\,s_g\big)\Big].
\end{equation}

\begin{proposition}\label{prop:cross}
For any $\sigma_1, \sigma_2 < 1$ and distinct products $j,k \in J_{hg}$, $\partial s_j/\partial\delta_k > 0$ if and only if
\begin{equation}\label{eq:condition}
s_{h|g}\big(\sigma_2 + (1-\sigma_2)\,s_g\big) < \frac{\sigma_2-\sigma_1}{1-\sigma_1}.
\end{equation}
The condition never holds if $\sigma_2 \le \sigma_1$ and $\sigma_1 \ge 0$, a set that contains the restricted parameter region. It holds whenever $s_{h|g}$ is small enough if $\sigma_1 < \sigma_2$ and $g$ has at least two subgroups, and whenever $s_g$ is small enough and $s_{h|g}$ large enough if $\sigma_2 \le \sigma_1 < 0$.
\end{proposition}

\begin{proof}
Since $\beta > 0$, the bracket in \eqref{eq:cross} is positive if and only if $X < R$, where $X = s_{h|g}\big(\sigma_2+(1-\sigma_2)s_g\big)$ and $R = (\beta-1)/\beta = (\sigma_2-\sigma_1)/(1-\sigma_1)$. Because $0 < s_g < 1$, the factor $\sigma_2+(1-\sigma_2)s_g = (1-s_g)\sigma_2 + s_g$ lies strictly between $\sigma_2$ and $1$. If $\sigma_2 \le \sigma_1$ and $\sigma_1 \ge 0$, then $R \le 0$ and $\sigma_2 - R = \sigma_1(1-\sigma_2)/(1-\sigma_1) \ge 0$, so $R \le \min(0,\sigma_2)$, while $0 < s_{h|g} \le 1$ gives $X > \min(0,\sigma_2)$; hence $X > R$. If $\sigma_1 < \sigma_2$, then $0 < R < 1$ and $X < s_{h|g}$, so $X < R$ whenever $s_{h|g} \le R$. If $\sigma_2 \le \sigma_1 < 0$, then $X \le s_{h|g}\sigma_2 + s_g(1-\sigma_2)$, and $s_{h|g}\sigma_2 < R$ whenever $s_{h|g} > (\sigma_1-\sigma_2)/(-\sigma_2(1-\sigma_1))$, a bound below one because $\sigma_1 < 0$; for such $s_{h|g}$, $X < R$ whenever $s_g < (R - s_{h|g}\sigma_2)/(1-\sigma_2)$.
\end{proof}

An improvement in product $k$, such as a price cut, raises the inclusive value of subgroup $h$ and draws consumers into it at the subgroup stage; once there, some of them draw product shocks that favor $j$. A marginal increase in $\delta_k$ raises the inclusive value of $h$ by $s_{k|hg}$, so $\log s_{h|g}$ rises by $(1-s_{h|g})\,s_{k|hg}/(1-\sigma_2)$, while $\log s_{j|hg}$ falls by $s_{k|hg}/(1-\sigma_1)$. When $\sigma_2 > \sigma_1$ and $s_{h|g}$ is small enough, the rise exceeds the fall: the inflow of consumers into $h$ outweighs the diversion from $j$ within it, and the rise in the share of group $g$ adds to the effect. For $k$ in another subgroup of $g$, $\partial s_j/\partial\delta_k = -(s_j s_k/s_g)\big[s_g + \sigma_2/(1-\sigma_2)\big] > 0$ if and only if $s_g < -\sigma_2/(1-\sigma_2)$, which requires $\sigma_2 < 0$.\footnote{For $k$ in another group or $j = 0$, $\partial s_j/\partial\delta_k = -s_j s_k$.} Under the sequential process, these patterns arise with single-unit demand, from the order in which consumers observe their shocks.

\section{Conclusion}\label{sec:conclusion}

The nested logit share system and its linear inversion arise exactly from a sequential choice process with i.i.d.\ extreme value shocks revealed in stages. This framework is consistent with expected utility maximization for any nesting parameters below one and any number of levels.

\section*{Acknowledgements}
I am grateful to Conor Ryan for his guidance on the two-level nested logit framework, to the Editor, Joseph E. Harrington, and to an anonymous referee for helpful comments.

\section*{Funding}
This research did not receive any specific grant from funding agencies in the public, commercial, or not-for-profit sectors.

\section*{Data availability}
No data was used for the research described in the article.

\section*{Declaration of generative AI and AI-assisted technologies in the manuscript preparation process}
During the preparation of this work the author used Claude (Anthropic) to search and check the literature and to draft and edit the text. After using this tool, the author reviewed and edited the content as needed and takes full responsibility for the content of the publication.

\bibliography{\bib}

\end{document}